\documentclass[letterpaper, 11 pt, conference]{ieeeconf}
\IEEEoverridecommandlockouts
\usepackage{cite}
\usepackage{amsmath,amssymb,amsfonts}
\usepackage{algorithm}
\usepackage{algorithmicx}
\usepackage{algpseudocode}
\usepackage{graphicx}
\usepackage{textcomp}
\usepackage{xcolor}
\usepackage{balance}
\usepackage{multirow}

\newtheorem{definition}{Definition}

\newtheorem{proposition}{Proposition}

\newtheorem{remark}{Remark}
\newtheorem{problem}{Problem}

\newcommand{\Xmu}{\mathcal{X}^\mu}
\newcommand{\Zeps}{Z_\epsilon}
\newcommand{\R}{\mathbb{R}}
\newcommand{\IB}{\mathbb{IB}}

\title{\LARGE \bf
Fast Constraint Extraction for Corrective Control under STL Specifications via Logical Dependency Tracking
}

\author{Antoine Besset$^{a}$, Joris Tillet$^{b}$, Chuchu Fan$^{c}$ and Julien Alexandre dit Sandretto$^{a}$
\thanks{$^{a}$U2IS, ENSTA, Institut Polytechnique de Paris, France; $^{b}$ Univ. Grenoble Alpes, CNRS, Grenoble INP, VERIMAG, France; $^{c}$ REALM, AeroAstro, Massachusetts Institute of Technology, Cambridge, USA}}

\begin{document}
\maketitle
\thispagestyle{empty}
\pagestyle{empty}

\begin{abstract}
Ensuring the satisfaction of Signal Temporal Logic (STL) specifications under uncertainty is challenging, as reachability-based monitoring provides guarantees but does not indicate how to restore satisfaction when it becomes indeterminate. A key difficulty is identifying which uncertain components actually affect global satisfaction, especially for nested formulas. This paper introduces a logical dependency tracking framework that propagates uncertainty through the STL structure and captures the causal contribution of reachable sets to satisfaction. By associating markers to uncertain predicates and propagating them via three-valued semantics, we extract in milliseconds a compact Disjunctive Normal Form (DNF) of sufficient constraints, avoiding combinatorial enumeration. As an application, we formulate control correction as a minimum-effort optimization problem. Using zonotopic reachability, the derived constraints are enforced via linear programming, yielding corrections that guarantee STL satisfaction under bounded uncertainty and provide certified probabilistic bounds in the stochastic case. We demonstrate the approach on a nonlinear system with nested STL specifications, showing that dependency tracking enables efficient and formally guaranteed correction. The tracking implementation is available at https://github.com/Antoine-Bst/STL-Three-Valued-Clause-Filtering/.
\end{abstract}

\section{Introduction}\label{sec:intro}
Ensuring the reliability of cyber-physical systems under uncertainty is a central concern in domains such as autonomous systems, robotics, and industrial automation. In such settings, model inaccuracies and external disturbances lead to deviations from nominal behavior, requiring methods that provide formal guarantees despite uncertainty. Reachability analysis~\cite{alexandre_dit_sandretto_validated_2016, Althoff2015ARCH, capdarticle} enables sound over-approximations of all possible trajectories, while Signal Temporal Logic (STL)~\cite{maler_monitoring_2004} provides a rich formal language to specify temporal and state-dependent requirements. Their combination enables the verification of temporal properties over continuous-time nonlinear systems with formal guarantees~\cite{Lercher2024CAV, Roehm2016ATVA, besset2025cdc}.

However, verification alone is not sufficient in safety-critical settings: when uncertainty leads to indeterminate specifications, the system must be corrected. While recent works provide guaranteed monitoring of STL specifications under uncertainty~\cite{Lercher2024CAV, besset2025cdc} or probabilistic assessment~\cite{Tran2025ProbStarTL}, they do not address how to compute corrective actions restoring satisfaction. Temporal implicants\cite{ferrereimplicant} explain a verdict on a trace, whereas under uncertainty satisfaction depends on a complex relation between temporal and logical operators, and indeterminacy typically arises when reachable sets intersect predicate boundaries. Crucially, not all such local ambiguities contribute equally to the global satisfaction.

In this paper, we assume the availability of a nominal trajectory already satisfying the specification, which must provide sufficient control margin for correction. Existing methods can compute such a trajectory, for instance via differentiable planning~\cite{dawson2022}, learning-based methods~\cite{meng2025telograf}, stochastic/tube-based optimization~\cite{sadigh2016safe, Farahani2018, safeautonomy2018, lindemann2021reactive, Yuonline2024}, robust control synthesis~\cite{verhagen2024}, or multi-agent synthesis~\cite{vlahakis2024probabilistic}. These methods compute a nominal trajectory in seconds to hundreds of seconds~\cite{Farahani2018, vlahakis2024probabilistic, lindemann2021reactive}, typically relying on discrete-time linear models~\cite{Farahani2018, safeautonomy2018, vlahakis2024probabilistic,Yuonline2024} or distributional assumptions on uncertainty~\cite{sadigh2016safe}. Planning such a trajectory under STL specifications with formal guarantees is intractable for continuous-time nonlinear systems: no closed-form solution is generally available for stochastic differential equations with nonlinear or hybrid dynamics. Moreover, on long horizons, reachability methods~\cite{alexandre_dit_sandretto_validated_2016, Althoff2015ARCH, capdarticle} accumulate over-approximation, which can render the resulting enclosures unusable for planning.
 
Rather than full replanning, we identify corrections online on top of a nominal controller, acting as a guaranteed correction layer between the high-level planner and the low-level feedback. The core contribution of this work is a logical dependency tracking framework that identifies which reachable sets, among those crossing predicate boundaries, are actually responsible for formula-level indeterminacy. Instead of relying on a direct All-SAT formulation, which is intractable due to combinatorial enumeration, we exploit the structure of three-valued satisfaction propagation to extract these dependencies, yielding a compact Disjunctive Normal Form of sufficient correction clauses (Section~\ref{sec:Results}). As an application of this formulation, we instantiate a control correction using zonotopic control parameterization, which interfaces naturally with existing set-based monitoring methods~\cite{Lercher2024CAV, besset2025cdc, shurmannzonocontrol2021, SCHURMANN2017}. Here, we target safety-critical systems, which motivates the use of a fully guaranteed approach: continuous-time semantics, nonlinear dynamics, and continuous-time validated enclosures.

\textbf{Related works.}
Several approaches exist for monitoring uncertain systems against STL specifications with formal guarantees, but none combines guaranteed satisfaction assessment with corrective control. Set-based methods evaluate STL formulas over reachable tubes using three-valued~\cite{Ishii2016IEICE, besset2025cdc} or four-valued~\cite{Lercher2024CAV} semantics, where indeterminate satisfaction arises when a reachable set intersects a predicate boundary. These methods provide sound verdicts but are limited to bounded disturbances. On the stochastic side, Monte-Carlo methods~\cite{SankarMonteCarlo} and Conformal prediction~\cite{lindemann2023conformal} estimate satisfaction probabilities but offer only statistical guarantees and may miss rare events. ProbStarTL~\cite{Tran2025ProbStarTL} tracks distributional dependencies through probabilistic zonotopes. For unbounded uncertainties such as Gaussian noise, we adopt a Potential Cloud approach~\cite{alexandre_dit_sandretto_confidence_2021} close to ProbReach~\cite{Shmarov2015ProbReach}, yielding a certified lower bound.
However, none of these approaches provides a mechanism to identify which uncertain components must be corrected to restore satisfaction.

On the control side,~\cite{shurmannzonocontrol2021, SCHURMANN2017} showed that zonotopic generators can parameterize guaranteed control corrections for constrained nonlinear systems, with the effect of each correction explicitly tracked through reachability. Funnel-based approaches~\cite{majumdar2017funnel} offer an alternative for controlling all states in tubes under disturbances. Since zonotopic reachability is already required for online monitoring, we build on~\cite{shurmannzonocontrol2021, SCHURMANN2017} and extend it to STL. The constraint identification method we develop is however independent of the control synthesis and applies to any approach capable of computing corrections under time-varying constraints.

\textbf{Contributions.} \textbf{(i)}~Extending prior work~\cite{besset2025cdc}, which identifies a conservative set of uncertain markers but not their logical dependencies, we extract these dependencies during signal propagation, yielding in milliseconds a compact Disjunctive Normal Form (DNF) of sufficient clauses without combinatorial enumeration, and scaling to hundreds of markers. This approach is compared against an All-SAT and single enumeration evaluation strategy. \textbf{(ii)} We propose a guaranteed control correction formulation based on zonotopic reachable tubes~\cite{shurmannzonocontrol2021, SCHURMANN2017}, 
where a single LP per clause yields a correction guaranteeing formula satisfaction.

\section{Preliminaries}\label{sec:prelim}

\label{sec:Preliminaries}
\subsection{Reachability analysis}
\label{sec:Reachability Analysis}
Let us consider a continuous dynamical system modeled by a differential equation of the form:
\begin{equation}
\dot{y}(t) = f(y(t), p), \quad y(t) \in \mathbb{R}^n,\ p \in \mathcal{P},
\end{equation}
where $y(t)$ denotes the system state, $p$ is a vector of parameters of the system, $\mathcal{P} \subset \mathbb{R}^w$ is a compact set. $f$ is a nonlinear function, such that for any initial state $y_0 \in \mathbb{R}^n$ and any
fixed parameters $p$, the system admits a unique trajectory in $\mathbb{R}^n$ denoted by $y(\cdot, y_0, p)$~\cite{alexandre_dit_sandretto_validated_2016,Althoff2015ARCH}.

In the process of solving an Initial Value Problem for a differential equation using validated numerical integration methods~\cite{alexandre_dit_sandretto_validated_2016}, an enclosure \([\tilde{y}_j]\) of the solution is computed for each time interval \([t_j, t_{j+1}]\), where \(t_{j+1} = t_j + h_j\) and \(h_j\) is an adaptive step size.
This continuous-time abstraction will be denoted as a Picard box. This enclosure satisfies:
\begin{definition}[Reachable tube enclosure~\cite{alexandre_dit_sandretto_validated_2016, Althoff2015ARCH, capdarticle}]
\begin{equation}
\mathrm{Reach}_{[t_j,t_{j+1}]}(\mathcal{Y}_0,\mathcal{P}) \subseteq [\tilde{y}_j]_{\mathcal{P}}, \bigcup_{j=0}^{N-1} [t_j,t_{j+1}] = [t_0,t_N],
\label{eq:reachablesetenclosure}
\end{equation}
where $\mathrm{Reach}_t(\mathcal{Y}_0,\mathcal{P}) = \{ y(t,y_0,p)\mid y_0\in\mathcal{Y}_0,\ p\in\mathcal{P} \},$
and $t_0$ and $t_N$ denote the start and end times of the simulation, respectively.
\end{definition}

This continuous-time representation, referred to as a tube $[\tilde{y}]_{\mathcal{P}}(t)$ for $t \in [t_0, t_N]$, is essential to preserve system behavior that may be lost through sparse time sampling.

\paragraph*{Notation}
In the following sections, several set operations are used: the Minkowski sum, denoted $\oplus$, and the Pontryagin difference, denoted $\ominus$.

\subsection{Signal temporal logic}
\label{sec:STL}
STL is an extension of Linear Temporal Logic designed to specify properties of continuous-time systems~\cite{maler_monitoring_2004}.
The syntax of STL is defined recursively as follows:
\begin{equation}
\phi := \mu \mid T \mid \neg \phi \mid \phi_1 \vee \phi_2 \mid \phi_1 \, U_{[a,b]} \, \phi_2,
\end{equation}
where $\phi$ is an STL formula, $T$ denotes the tautology, $\neg$ the negation, and $\mu$ is an atomic predicate specifying signal constraints (e.g., $x > 0$)~\cite{maler_monitoring_2004}.
A trace $(s,t)$ satisfies a formula at $t$ iff $(s,t)\models \phi$.
Derived operators are defined as follows: 
\begin{equation}
\begin{aligned}
     &F_{[a,b]} \phi = T U_{[a,b]} \phi, \quad G_{[a,b]} \phi = \neg(F_{[a,b]} \neg\phi)\\
     & \phi_1 \land \phi_2 \equiv \neg (\neg \phi_1 \lor \neg \phi_2), \quad \phi_1 \implies \phi_2 \equiv \neg \phi_1 \lor \phi_2.
\end{aligned}
\end{equation}

In this work, we consider the fragment of STL composed of arbitrarily nested ${F}_{[a,b]}$ and ${G}_{[a,b]}$ operators, which covers the specifications commonly encountered in robotic and control applications (reach-avoid, recurrence, sequencing, $\dots$).

\subsection{Three-Valued STL on Tubes}
We evaluate atomic predicates on reachable sets using Boolean intervals $\IB = \{0, 1, [0,1]\}$~\cite{besset2025cdc,Lercher2024CAV}:
\begin{equation}\label{eq:threeval}
    ([\tilde{y}]_{\mathcal{P}}, t) \models \mu := \begin{cases}
        1, & \text{if } [\tilde{y}]_{\mathcal{P}}(t) \subset \Xmu, \\
        0, & \text{if } [\tilde{y}]_{\mathcal{P}}(t) \cap \Xmu = \emptyset, \\
        {[0,1]}, & \text{otherwise},
    \end{cases}
\end{equation}
where $\Xmu$ is the predicate set, represented as a hyperbox to facilitate the formulation as a linear optimization problem.
Here $[0,1] \in \IB$ denotes the indeterminate value, not a real interval.
Logical operations are extended via Boolean interval arithmetic~\cite{jaulin_applied_2001}:
\begin{definition}[Boolean interval arithmetic~\cite{jaulin_applied_2001}]
\begin{equation}
  \begin{aligned}
[a] \vee [b] &= \{a \vee b \mid a \in [a], b \in [b]\}\\
\neg[a] &= \{\neg a \mid a \in [a]\}.
\end{aligned}
\label{def:booleaninterval}
\end{equation}
\end{definition}
\begin{align*}
\text{i.e. } 0 \wedge [0,1] = 0&, \quad 0 \vee [0,1] = [0,1],\\
1 \wedge [0,1] = [0,1]&, \quad 1 \vee [0,1] = 1.
\end{align*}
Satisfaction signals $[\mathcal{S}_\phi](t)$ are decomposed into certain and uncertain unitary signals, propagated bottom-up through the STL formula tree~\cite{besset2025cdc}.

For a bounded parameter set $\mathcal{P}$, the outcome is either:
\begin{equation}\label{eq:tubeverdict}
([\tilde{y}]_{\mathcal{P}}, t) \models \phi = 
\begin{cases}
    1, & \text{all trajectories satisfy } \phi, \\
    0, & \text{all trajectories violate } \phi, \\
    {[0,1]}, & \text{indeterminate.}
\end{cases}
\end{equation}

\subsection{Guaranteed probabilistic bound}\label{sec:probabound}
Similarly to ProbReach~\cite{Shmarov2015ProbReach}, which provides a bounded assessment of the probability of satisfaction, we propose:
\begin{proposition}
Let $\mathcal{S}_{c} \subseteq \mathcal{P}$ be a measurable set of stationary independent parameters with probability
$\Pr(\hat{p} \in \mathcal{S}_{c}) = c$. Then the random trajectory from $\hat{p}$, $\hat y(\cdot)$ follows:
\begin{equation}
\boxed{
\begin{aligned}
    \mathcal{S}_{c} \text{ is a validation set}
    &\;\Longrightarrow\;
    \Pr\big((\hat y, t) \models \phi\big) \ge c, \\[4pt]
    \mathcal{S}_{c} \text{ is an invalidation set}
    &\;\Longrightarrow\;
    \Pr\big((\hat y, t) \not\models \phi\big) \ge c .
\end{aligned}
}
\end{equation}
\label{prop:lowerbound}
\end{proposition}

\textit{Proof :}
Let, $\mathcal{P}_\phi = \left\{\, p \;\middle|\; (y(p), t) \models \phi = 1 \,\right\}$, denote the set of parameter values satisfying the formula~$\phi$.
We assume the existence of a certified validation subset $\mathcal{S}_c \subseteq \mathcal{P}_\phi$ with a probability measure $c$.
By monotonicity of the measure, for a random parameter $\hat p$ and $\hat y(\cdot) := y(\hat p, \cdot)$, it follows that:
\begin{equation}
   \Pr\big(\hat{p} \in \mathcal{S}_c\big) 
\leq 
\Pr\big(\hat{p} \in \mathcal{P}_\phi\big) 
= 
\Pr\big((\hat y, t) \models \phi\big).
\label{eq:safesetineq}
\end{equation}
Let $\phi' = \neg \phi$. Since $(y(p), t) \models \phi' \iff (y(p), t) \not\models \phi$, the monotonicity property extends to the invalidation set. $\blacksquare$

In other words, the probability of a validated parameter set lower-bounds the satisfaction probability.

\paragraph*{Probabilistic Notation}
$\mathbb{P}(S_c) = c$ denotes the probability measure of a measurable set $S_c \subset \Omega$, where $\Omega$ is the sample space with its measure function $\mathbb{P}$.

\subsection{Problem Formulation}

Consider a nonlinear uncertain system
\begin{equation}\label{eq:system}
    \dot{y}(t) = f(y(t), u(t), p), \quad y(0) \in Y_0, \quad p \in \mathcal{P} \subset \R^w,
\end{equation}
where the uncertain parameters $p_i$ follow known probability distributions at $t = 0$ and are assumed independent and stationary. No assumption is made on the state distribution of $\hat{y}$ at $t > 0$.
The control input is piecewise constant over $N_c$ intervals $[t_j, t_{j+1}], ~ j \in \{1, \dots, N_c\}$:
\begin{equation}\label{eq:control}
    u(t) =   U_{\mathrm{nom},j} + \sum_{n=1}^m \delta U_{j,n} \cdot v_{j,n} \quad \text{if } t \in [t_j, t_{j+1}],
\end{equation}
$U_{\mathrm{nom},j} = U_{ff,j} + U_{fb,j}(t) \in \mathbb{R}^m$ is the nominal command composed of a feedforward and a feedback term linearized at each step $j$ and
designed to satisfy the specification at $p = \mathbb{E}[p]$ (or the center of a bounded parameter set). For each input, a correction is applied, using $m$ orthogonal $\delta U_{j,n} \in \mathbb{R}^m$ encoding an amplitude. A corresponding decision variable $v_{j,n} \in [-1, 1]$ is also associated. All decision variables are flattened in a decision vector $V \in \mathbb{R}^{m\cdot N_c}$ with coordinate $v$.
The nominal trajectory computed from $U_\text{nom}$ is denoted $y_\text{ref}$.

\begin{problem}\label{prob:main}
Given a nominal trajectory satisfying an STL specification $\phi$ and a target confidence level $c \in [0,1]$, find the minimum-effort control correction guaranteeing $\Pr((\hat{y}, t) \models \phi) \geq c$:
\begin{equation}\label{eq:problembox}
\boxed{
\begin{aligned}
    \min_{v} \quad & \|V\|_1 \\
    \text{s.t.} \quad 
    & \Pr\bigl((\hat{y}, t) \models \phi\bigr) \geq c \\
    & v_j \in [-1, 1], \quad j \in \{1, \ldots, m \cdot N_c\}
\end{aligned}
}
\end{equation}
\end{problem}

Hence for bounded uncertainties: $\mathrm{s.t} ~ ({y}, t) \models \phi, \forall p \in \mathcal{P}$.

The chance constraint $\Pr((\hat{y}, t) \models \phi) \geq c$ is not directly tractable. The next sections reduce it to inclusion conditions on the reachable tube at critical time intervals, which are then enforced as linear constraints.

\section{Critical Constraints Identification}\label{sec:CriticalConstraints}

Translating a temporal formula into constraints requires identifying where the system must satisfy strict inclusion or exclusion. This is non-trivial due to the infinite choice of satisfaction instants under quantifiers such as $\exists t \in [a,b]$, and the indeterminacy from predicate crossings. We formulate these critical constraints for any STL formula (Section~\ref{sec:STL}) by tracking logical dependencies in continuous time.

\subsection{Mitigating Uncertainty over Subsets of Markers}
When evaluating an STL formula on a reachable tube, atomic predicates are assessed over time intervals defined by Picard boxes. The goal is to identify where the corrected tube must satisfy strict predicate inclusion or exclusion.

Following~\cite{besset2025cdc}, we propagate three-valued satisfaction signals through the STL formula tree. Uncertain time intervals are annotated with \emph{markers} pointing to the reachable sets responsible for indeterminate satisfaction as well as the associated predicates.

We define the uncertainty mitigation clause $Q_j$ at atomic predicate level as follows:
\begin{equation}
    Q_j \;\iff \; \bigl( [\tilde{y}_j]_{\mathcal{S}_c} \subset \mathcal{X}^\mu \bigr) \;\lor\; \bigl( [\tilde{y}_j]_{\mathcal{S}_c} \cap \mathcal{X}^\mu = \emptyset \bigr).
    \label{eq:predicateclause}
\end{equation}

\begin{definition}[Marker]
A marker $\mathcal{M}$ over time interval $I$ is a logical expression defined recursively as either:
\begin{itemize}
    \item a predicate clause $Q_j$ (leaf), or
    \item a logical tree with a combination $\mathcal{M}_A \lor \mathcal{M}_B$ or $\mathcal{M}_A \land \mathcal{M}_B$ of two markers (node).
\end{itemize}
Each marker is associated with an expected satisfaction value $E(\mathcal{M}) \in \{0, 1\}$ derived from the nominal trajectory.
\end{definition}

From the method proposed in~\cite{besset2025cdc}, access to the reachable sets leading to undetermined satisfaction is sufficient to mitigate indeterminacy.
However, this approach is too pessimistic, as only subsets of these markers usually suffice. In the following, we compute a DNF of sufficient correction clauses.

Let us define a master formula $\psi \in \{0,1\}$, a meta-Boolean encoding whether satisfaction of $\phi$ is guaranteed:
\begin{definition}[Master formula]
\begin{equation}
    \psi := \begin{cases}
    1, & \text{if } ([\tilde{y}]_{S_c}, t) \models \phi = 1,\\
    0, & \text{if } ([\tilde{y}]_{S_c}, t) \models \phi = [0,1].
\end{cases}
\end{equation}
\end{definition}

In our formulation, since the correction is computed along a reference trajectory, constraints are directly enforced in the clauses $Q_j$, whether they correspond to inclusion or exclusion conditions. The DNF becomes:
\begin{equation}
\psi \equiv \bigvee_{k\in\mathcal{K}} \bigwedge_{j\in C_k} Q_j,
\label{eq:dnfclause}
\end{equation}
where $\mathcal{K}$ is an index set and each $C_k$ denotes the subset of clauses $Q_j$ appearing in the $k$-th conjunction of the DNF. Each conjunction is logically minimal: removing any clause $Q_j$, $j \in C_k$, leaves $\psi = 0$.

However, this approach suffers from several limitations. First, the formula is monitored in continuous time, and the Picard boxes are abstractions over time intervals of varying lengths, propagated using the logic introduced in~\cite{maler_monitoring_2004, besset2025cdc}. Second, rewriting the formula in discrete time under a DNF using an LTL-based approach loses continuous-time guarantees and introduces a combinatorial complexity that can become intractable~\cite{Tran2025ProbStarTL}, particularly when online implementation is required.

\subsection{Propagation Rules to Track Logical Dependencies}
\label{sec:tracking}

We compute logical dependencies directly during marker propagation through the STL syntax tree, by discarding insufficient clauses and constructing a logical expression when necessary.
Following~\cite{maler_monitoring_2004, besset2025cdc}, satisfaction signals are decomposed into certain and uncertain unitary parts, each carrying an expected value derived from the nominal trajectory (e.g., $E(A)=1$ if satisfaction is $1$ on time interval $I_A$). Markers are propagated by exploiting dominance relations between expected values, as summarized in Table~\ref{tab:orpropag}, \ref{tab:andpropag}, \ref{tab:negpropag}, where (1) and (0) denote certain satisfaction.

\begin{table}[h]
\centering
\caption{Marker propagation for $C = A \lor B$}
\begin{tabular}{cccccc}
\hline
$\mathcal{M}_A$ & $E(A)$ & $\mathcal{M}_B$ & $E(B)$ & $\mathcal{M}_C$ & $E(C)$ \\
\hline
$\checkmark$ & 0 & $\times$ & (0) & $\mathcal{M}_A$ & 0 \\
$\checkmark$ & 1 & $\times$ & (0) & $\mathcal{M}_A$ & 1 \\
$\times$ & (0) & $\checkmark$ & 0 & $\mathcal{M}_B$ & 0 \\
$\times$ & (0) & $\checkmark$ & 1 & $\mathcal{M}_B$ & 1 \\
$\checkmark$ & 1 & $\checkmark$ & 0 & $\mathcal{M}_A$ & 1 \\
$\checkmark$ & 0 & $\checkmark$ & 1 & $\mathcal{M}_B$ & 1 \\
$\checkmark$ & 1 & $\checkmark$ & 1 & $\mathcal{M}_A \lor \mathcal{M}_B$ & 1 \\
$\checkmark$ & 0 & $\checkmark$ & 0 & $\mathcal{M}_A \land \mathcal{M}_B$ & 0 \\
\hline
\end{tabular}
\label{tab:orpropag}
\end{table}

\begin{table}[h]
\centering
\caption{Marker propagation for $C = A \land B$}
\begin{tabular}{cccccc}
\hline
$\mathcal{M}_A$ & $E(A)$ & $\mathcal{M}_B$ & $E(B)$ & $\mathcal{M}_C$ & $E(C)$ \\
\hline
$\checkmark$ & 0 & $\times$ & (1) & $\mathcal{M}_A$ & 0 \\
$\checkmark$ & 1 & $\times$ & (1) & $\mathcal{M}_A$ & 1 \\
$\times$ & (1) & $\checkmark$ & 0 & $\mathcal{M}_B$ & 0 \\
$\times$ & (1) & $\checkmark$ & 1 & $\mathcal{M}_B$ & 1 \\
$\checkmark$ & 0 & $\checkmark$ & 1 & $\mathcal{M}_A$ & 0 \\
$\checkmark$ & 1 & $\checkmark$ & 0 & $\mathcal{M}_B$ & 0 \\
$\checkmark$ & 0 & $\checkmark$ & 0 & $\mathcal{M}_A \lor \mathcal{M}_B$ & 0 \\
$\checkmark$ & 1 & $\checkmark$ & 1 & $\mathcal{M}_A \land \mathcal{M}_B$ & 1 \\
\hline
\end{tabular}
\label{tab:andpropag}
\end{table}

\begin{table}[h]
\centering
\caption{Marker propagation for $D = \neg C$}
\begin{tabular}{cccc}
\hline
$\mathcal{M}_C$ & $E(C)$ & $\mathcal{M}_D$ & $E(D)$ \\
\hline
$\checkmark$ & 1 & $\mathcal{M}_C$ & 0 \\
$\checkmark$ & 0 & $\mathcal{M}_C$ & 1 \\
\hline
\end{tabular}
\label{tab:negpropag}
\end{table}

All reported propagation rules demonstrate the computation of an uncertain satisfaction signal ($C$ or $D$) over a time interval $I$. A checkmark ($\checkmark$) indicates that a marker is present in the uncertain signal, while a cross ($\times$) denotes a certain signal carrying no marker. 

When a new uncertain signal is computed, dominance relations are exploited: for $\lor$, the marker with $E = 1$ dominates, as it alone determines satisfaction. Similarly, for $\land$ the marker with $E = 0$ dominates~\cite{jaulin_applied_2001}. In the symmetric cases where both expected values agree, the resulting marker captures the joint dependency as $\mathcal{M}_A \lor \mathcal{M}_B$ or $\mathcal{M}_A \land \mathcal{M}_B$ accordingly. Finally, negation simply inherits the marker while flipping the expected value (due to the symmetry $\neg [0,1] = [0,1]$).

For example, in Table~\ref{tab:orpropag}, line 6, consider $A \lor B$ with $E(A)=0$ and $E(B)=1$. Mitigating $A$ alone leaves the satisfaction at $[0,1]$, since $[0,1] \lor 0 = [0,1]$. Only mitigating $B$ resolves the uncertainty: $[0,1] \lor 1 = 1$. Here $B$ dominates the satisfaction.

According to these propagation rules, evaluating a formula at a given time yields a mostly flat logical structure between markers. This keeps the DNF computation tractable even while considering 50 to 700 markers as confirmed by the millisecond-range computation times reported in Section~\ref{sec:Results}.

\begin{figure}[htbp]
   \begin{center}
    \includegraphics[width=\linewidth, trim=0 8 0 5, clip]{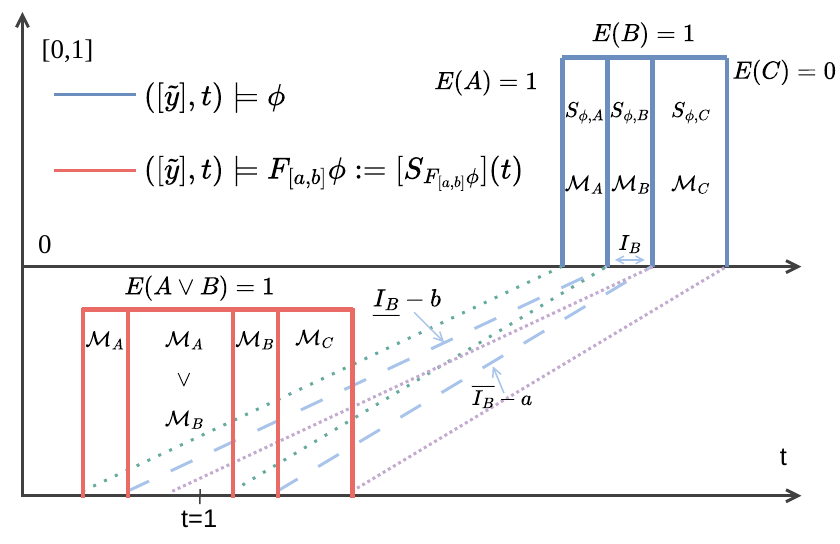}
      \caption{Computing the satisfaction signal of $F_{[a,b]}\phi$ under uncertainty and marker propagation, shift-back in dashed line, x and y axis is time and satisfaction respectively. $I_A$ and $I_B$ are the time intervals defining each unitary signal $S_{\phi,A}$, $S_{\phi,B}$.}
      \label{satmarker}
      \end{center}
   \end{figure}
%
   
An example of marker propagation is shown in Fig.~\ref{satmarker}. An uncertain signal for a formula $\phi$ is decomposed into unitary parts, each associated with a marker. Each part $S_{\phi,i}$ is shifted back by $\ominus [a,b]$ to compute the satisfaction signal of $ {F}_{[a,b]}\phi$:
\begin{equation}
[S_{  {F}_{[a,b]}\phi}](t) = \bigvee_{i \in \{A,B,C\}} \left(S_{\phi,i} \ominus [a,b]\right)(t),
\end{equation}
following~\cite{maler_monitoring_2004, besset2025cdc}. For example, $\mathcal{M}_A = Q_1 \land Q_2 \land Q_3$ and $\mathcal{M}_B = Q_4 \land Q_5$, yielding the DNF at $t=1$:
\begin{equation}
\psi = (Q_1 \land Q_2 \land Q_3) \lor (Q_4 \land Q_5).
\end{equation}
Guaranteed satisfaction of $\phi$ thus requires mitigating either clause $Q_1 \land Q_2 \land Q_3$ or $Q_4 \land Q_5$. This selection of clause and their expected value are the necessary constraints that the system should meet under control correction. In the implementation, the resulting logical expression is converted to DNF using standard Boolean algebra transformations, duplicated clauses being merged.
While~\cite{besset2025cdc} treats these markers as jointly necessary, the propagation rules of Tables~\ref{tab:orpropag}--\ref{tab:negpropag} recover their logical relations, exposing \emph{alternative} sufficient corrections rather than a single conservative set, each conjunction being logically minimal since dominated markers are discarded.

\begin{remark}[Reference trajectory]
In this example, the signal part associated with $\mathcal{M}_C$ yields a satisfaction of $0$ regardless of the correction, since the nominal trajectory has an expected value of $0$. If the formula requires a $1$ satisfaction over the corresponding time interval, the initial trajectory is therefore ill-defined and well identified.
\end{remark}

\subsection{Inherited Constraints}\label{sec:identify}
\subsubsection{Problem formulation}
The key insight is that the exact position of the corrected tube with respect to the nominal trajectory is irrelevant: it suffices to guarantee that the entire reachable tube lies inside the required predicate sets at the prescribed time instants or intervals, effectively reducing the problem to several coupled reach-avoid constraints.
The DNF formulation decomposes this condition into alternative candidate constraint sets, among which the least costly is selected.

Given a reachable tube computed from an initial parameter set $S_\text{init}$, a tube $[\tilde{y}]_{S_c, V}$, originating from $S_c \subset S_\text{init}$ and command corrections, the decision vector $V$ identification problem is formulated as:

\begin{equation}\label{eq:problemdnf}
\boxed{
\begin{aligned}
    \min_{v, k} \quad &  \|V\|_1 \\
    \text{s.t.} \quad 
    & [\tilde{y}_j]_{S_c, V} \models Q_j, \quad \forall j \in \mathcal{C}_k, \\
    & \mathcal{C}_k \in \mathrm{DNF}(\phi, y_\text{ref}, [\tilde{y}]_{S_\text{init}}),\\
    & [\tilde{y}_j]_{S_c, V} \models \mathrm{Cst}(\phi, \mathcal{C}_k) \\
    & v_n \in [-1, 1], \quad n \in \{1, \ldots, m \cdot N_c\}
\end{aligned}
}
\end{equation}
 where $\mathrm{Cst}(\phi, \mathcal{C}_k)$ collects the constraints inherited from the formula. Intuitively, these constraints capture all conditions that must remain satisfied once a correction is enforced.
 
 \subsubsection{Inherited Constraints}
 Resolving the uncertainty captured by the $\mathcal{C}_k$ clauses is \textbf{not sufficient}: the correction must also satisfy all implicit constraints inherited from the nominal trajectory (e.g., obstacle avoidance, timing of temporal operators).
 
\begin{definition}[Inherited constraints]
Let $\phi$ be an STL formula and let $\mathcal{C}_k \in \mathrm{DNF}(\phi, y_{\text{ref}}, [\tilde{y}]_{S_{\text{init}}})$ be a set of clauses such that the master formula $\psi = 1$. $Q'_j$ is a predicate clause on time interval from Eq.~\eqref{eq:reachablesetenclosure} and \eqref{eq:predicateclause}. The set of constraints is:
\begin{equation}
    \mathrm{Cst}(\phi, \mathcal{C}_k) := \{ Q'_j \mid (\bigwedge_{i \in \mathcal{C}_k} Q_i)\land \neg Q'_j \Rightarrow \psi = 0 \}.
\end{equation}
\end{definition}
 
To identify these, we fix the predicate clauses in $\mathcal{C}_k$ as certain based on their expected value, and relax all other atomic predicate signals to $[0,1]$ while keeping their expected value unchanged. The tracking algorithm of Section~\ref{sec:tracking} is then applied, potentially yielding hundreds to thousands of clauses. Portions whose indeterminacy does not imply formula-level indeterminacy are discarded. The overall identification remains in a tractable time, for example 20--50\,ms range as reported in Section~\ref{sec:Results}, owing to the efficiency of the logical tracking.

   \begin{figure}[htbp]
       \begin{center}
        \includegraphics[width=\linewidth, trim=10 0 80 0, clip]{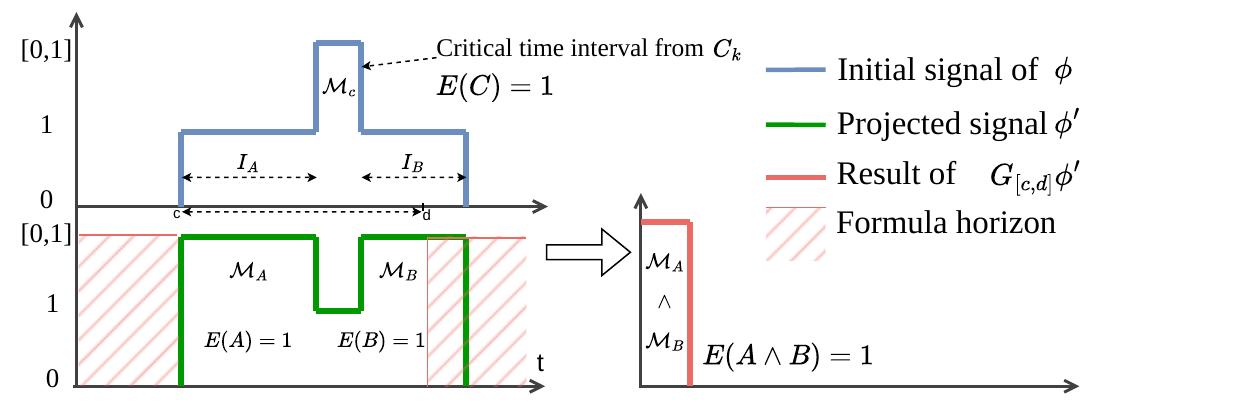}
      \caption{Initial signal with critical time in blue, $[0,1]$ projection in green, projection result in red with satisfaction relation.}
      \label{projection}
      \end{center}
   \end{figure}

The approach is illustrated in Fig.~\ref{projection}. The original tube exhibits uncertainty due to the $\mathcal{C}_k$ clause in $\mathcal{M}_C$. To derive the necessary constraints associated with this clause, the full signal is switched to $[0,1]$, and the satisfaction of the projected signal is evaluated while tracking timed and logical dependencies. This yields the additional necessary constraints $\mathcal{M}_A \land \mathcal{M}_B$ required to satisfy the full formula. Here $I_A$ and $I_B$ are time intervals corresponding to Picard boxes. The red-hatched regions correspond to intervals filtered out based on the formula horizon when evaluating at $t = 0$, this horizon is computable for any formula~\cite{Lercher2024CAV, maler_monitoring_2004}.

\begin{remark}[Optimality]
\label{rmk:optimality}
Obtaining the global minimum correction would require computing the DNF of constraints from $\mathrm{Cst}(\phi, \mathcal{C}_k)$, yielding $K'$ additional disjunctions and potentially $K_{\textrm{tot}} \approx K \cdot K'$ clauses to compute control correction. Here we compute this DNF and filter to one clause to keep the computation tractable, using a set-membership score $sc([\tilde{y}], \mathcal{X}^\mu) =\frac{\nu (\mathcal{X}^\mu \cap [\tilde{y}])}{\nu [\tilde{y}]}$, with $\nu$ the volume of the set.
\end{remark}

\subsection{Probabilistic Extension and Scaling Vector}\label{sec:grid}
The formulation in Eq.~\eqref{eq:problemdnf} naturally extends to the probabilistic setting of Proposition~\ref{prop:lowerbound}. Let $S_{\text{init}}$ be a mean-centered parameter set with $\mathbb{P}(S_{\text{init}})\gg c$. A scaling vector $A$ defines a subset $S_c \subseteq S_{\text{init}}$ such that $\mathbb{P}(S_c)=c$.

For $w\geq 2$, infinitely many such $A$ exist, inducing different distributions of uncertainty across parameters and thus different reachable tube widths. This flexibility yields distinct constraints on the control zonotopes $Z_{v,k}$, enabling the selection of a minimal correction (see Section~\ref{sec:Results}). The vector $A$ can be sampled over a discrete grid, and the corresponding sets propagated using contractor-based methods~\cite{alexandre_dit_sandretto_confidence_2021}.

\begin{remark}[Chance constraints]
The probabilistic constraint $\Pr((\hat y, t) \models \phi) \geq c$ of Problem~\ref{prob:main} is a chance constraint. Rather than solving it as a stochastic program~\cite{safeautonomy2018}, Proposition~\ref{prop:lowerbound} reduces it to a deterministic inclusion on the reachable tube of $S_c$.
\end{remark}

\section{Control Correction Using Zonotopes}\label{sec:ControlCorrection}
This section computes the control correction enforcing strict reachable set inclusion or exclusion.
\subsection{Command Parameterization}
We encode the control correction as zonotopic generators. On each time interval $[t_k, t_{k+1}],~ k \in \{1, \dots , N_c\}$, the corrected command is
\begin{equation}\label{eq:correctedcmd}
    u(t) =   U_{\mathrm{nom},k} + \sum_{n=1}^m \delta U_{k,n} \cdot v_{k,n}, ~ v_{k,n} \in [-1, 1],
\end{equation}
where $\delta U_{k,n} \in \R^m$ is a fixed correction direction and amplitude and $v_{k,n}$ is the decision variable. Each $v_{k,n}$ (flattened in $v_j$) is associated with a unique noise symbol $\varepsilon_{\delta,j}$ in the zonotopic representation. The resulting control zonotope has at most $m \cdot N_c$ generators, each symbol corresponding to a coordinate of decision vector $V$. 

After validated integration, each reachable set can be written as
\begin{equation}\label{eq:propagated}
\begin{aligned}
    [\tilde{y}]_{S_c, \delta U} =& X_0 ~ \oplus  ~ \Bigr\{\sum_{i} X_{A,i} \varepsilon_i \Big| \varepsilon_i \in [-1,1] \Bigr\} \\
    \oplus ~ &\Bigr\{\sum_{j} X_{\delta U_j} \varepsilon_{\delta,j} \Big| \varepsilon_{\delta,j} \in [-1,1] \Bigr\}~ \oplus~ \Zeps \\
    := & X_0 \oplus Z_A \oplus Z_v \oplus \Zeps,
\end{aligned}
\end{equation}
where $Z_A$ captures the uncertainty arising from the parameter set $S_c = S_\text{init}(A)$, $Z_v$ encodes the effect of each correction $v_j$ on the reachable set through the nonlinear dynamics, and $\Zeps$ aggregates the contributions of nonlinearities either from the control or uncertainties, the Picard box time enclosure, and numerical approximation errors.

The reachability solver tracks symbolic dependencies through the propagation. Fixing $\varepsilon_{\delta,j} = v_j^*$ shifts the center of the reachable set by $\delta y_k = \sum_{j} X_{\delta U_j, k} \cdot v_j^*$.

\subsection{Problem Formulation}
\subsubsection{Geometric Constraints}
\begin{proposition}
For each critical reachable set $k$ (at $[t_k, t_{k+1}]$) with an inclusion constraint in $\Xmu$, the problem can be formulated as an inclusion in a reduced predicate:
\begin{equation}\label{eq:geometriccst}
\begin{aligned}
    \Xmu_{A,\epsilon,k} &:= (\Xmu \ominus X_{0,k}) \ominus \Box({Z_\epsilon}_{,k} \oplus Z_{A,k}),\\
     &\text{and }Z_{v,k} \subseteq \Xmu_{A,\epsilon,{k}} \implies [\tilde{y}_k] \subset \Xmu,
\end{aligned}
\end{equation}
where $\Box(\cdot)$ denotes the interval hull. Note that the equivalence holds when $Z_A \oplus {\Zeps}$ is a box, which is the case after a necessary order reduction to compute the hull of the zonotope for predicate evaluation.
\end{proposition}

\textit{Proof: }
Starting from the full inclusion:
\begin{equation}
    X_{0,k} \oplus Z_{v,k} \oplus Z_{A,k} \oplus {Z_\epsilon}_{,k} \subseteq \Xmu,
\end{equation}
we have by translating $-X_{0,k}$:
\begin{equation}
    Z_{v,k} \oplus Z_{A,k} \oplus {Z_\epsilon}_{,k} \subseteq \Xmu \ominus X_{0,k}.
\end{equation}

Since $\Xmu$ is an axis-aligned box, $\Box({Z_\epsilon}_{,k})$ and $\Box(Z_{A,k})$ is a zero-centered symmetric interval hull~\cite{jaulin_applied_2001}:
\begin{equation}
\begin{aligned}
    &\Box(Z_{v,k}) \subseteq (\Xmu \ominus X_{0,k}) \ominus \Box({Z_\epsilon}_{,k} \oplus Z_{A,k}).\\
    &\implies \Box(Z_{v,k} \oplus Z_{A,k} \oplus {Z_\epsilon}_{,k}) \subseteq \Xmu \ominus X_{0,k} \\
    & \iff Z_{v,k} \oplus Z_{A,k} \oplus {Z_\epsilon}_{,k} \subseteq \Xmu \ominus X_{0,k}
\end{aligned}
\end{equation}

Each successive Pontryagin difference preserves the axis-aligned box structure. The resulting set $\Xmu_{A,\epsilon,k} := \Xmu \ominus X_{0,k} \ominus \Box({Z_\epsilon}_{,k} \oplus Z_{A,k})$ is a box (possibly $\emptyset$ if the uncertainty exceeds the predicate margin).
Then, 
\begin{equation}
    \Box(Z_{v,k}) \subseteq \Xmu_{A,\epsilon,k} \iff Z_{v,k} \subseteq \Xmu_{A,\epsilon,k} \implies
    [\tilde{y}_k] \subset \Xmu ~\blacksquare
\end{equation}

The inclusion of $Z_{v,k}$ in the reduced predicate set $\Xmu_{A,\epsilon,k}$ is equivalent to a set of linear constraints~\cite{jaulin_applied_2001}, which can be directly incorporated into the LP formulation. For a fixed scaling vector $A$, the minimum-effort control correction is obtained by solving the linear program
\begin{equation}\label{eq:lp}
\begin{aligned}
    \min_{v} \quad & \sum_{j=1}^{m \cdot N_c} |v_j| \\
    \text{s.t.} \quad & C_{\min,k} \leq G_k V \leq C_{\max,k}, \quad \forall k \in \mathcal{K}, \\
    & v_j \in [-1,1],
\end{aligned}
\end{equation}
where $G_k$ is the command generator matrix, $C_{\min,k}$ and $C_{\max,k}$ define the bounds of the reduced predicate, and $\mathcal{K}$ denotes the set of critical reachable sets associated with a selected clause $\mathcal{C}$ and its inherited constraints $\mathrm{Cst(\phi,\mathcal{C})}$ (see Eq.~\eqref{eq:problemdnf}). Exclusion constraints are handled similarly via half-space selection, reducing to an equivalent inequality.

\begin{remark}[Theoretical guarantee]
The theoretical guarantee of Eq.~\eqref{eq:geometriccst} holds if and only if each control correction remains within its admissible bounds: every coordinate of $V$, $v_j \in [-1, 1]$, as ${Z_\epsilon}_{,k}$ is computed assuming these bounds.
\end{remark}

\subsection{Feasibility and Iterative Refinement}
\label{sec:feasibility}
Infeasibility of the LP in Eq.~\eqref{eq:lp} does not necessarily imply infeasibility of the original problem, but may result from conservatism in $Z_\epsilon$ or tight bounds $v_j \in [-1,1]$. These effects are coupled, as relaxing bounds increases $Z_\epsilon$. 
An iterative scheme progressively tightens relaxed control bounds after recomputing the tube, though convergence is not guaranteed. Iterations are capped and persistent infeasibility triggers a failsafe fallback (replanning or safe stop).

\section{Application and Results}
The previous results are applied to an uncertain nonlinear vehicle. The nominal trajectory is chosen to present ambiguous cases that exercise the key features of our approach.
\label{sec:Results}

\begin{figure*}[th!]
    \begin{center}
        \includegraphics[width=\linewidth, trim=10 160 40 10, clip]{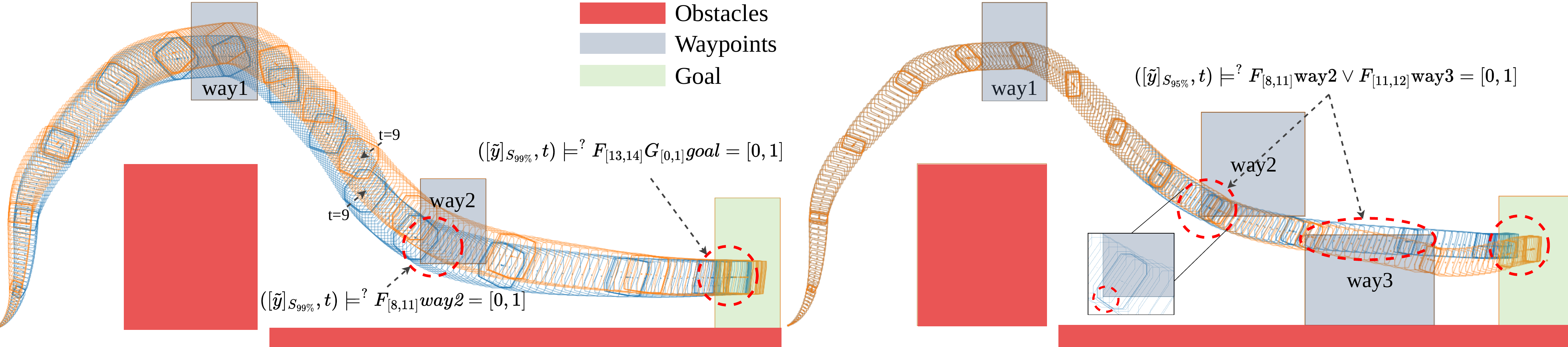}
        \caption{Initial tube in blue for $c=0.99$ (left) and $c=0.95$ (right), corrected tube in orange, red circles show indeterminacy regions.}
        \label{probadisturb}
    \end{center}
\end{figure*}

\subsection{Setup}
We consider a Dubins-like vehicle with dynamics
\begin{equation}
\begin{aligned}
    \dot{x} &= v\cos\theta, \quad \dot{y} = v\sin\theta, \\
    \dot{v} &= K_v(u_v - v), \quad \dot{\theta} = K_c(u_c - \theta),
\end{aligned}
\end{equation}
where $(u_v,u_c)^\top = U_{\mathrm{nom},j} = U_{ff,j} + U_{fb}(x,y,x_{\mathrm{ref},j},y_{\mathrm{ref},j})$, and $K_v$, $K_c$ are the throttle and steering gains. The piecewise command has $N_c=15$ steps, yielding $2N_c=30$ decision variables. Uncertainties affect $K_v$, $K_c$, and additive perturbations $P$, leading to
\begin{equation}
    \dot{v} = K_v(u_v - v) + P_v, \qquad \dot{\theta} = K_c(u_c - \theta) + P_c.
\end{equation}

The STL specification involves five predicates: two waypoints \textit{way1} and \textit{way2}, two obstacles \textit{obs1} and \textit{obs2}, and a \textit{goal} region:
\begin{equation}
\begin{aligned}
    \phi & = \underbrace{\left( F_{[5,7]} \textit{way1} \land F_{[8,11]} \textit{way2} \land G_{[0,t_f]} (\neg \textit{obs1} \land \neg \textit{obs2}) \right)}_{\text{waypoints \& obstacle avoidance (\textbf{flat} formula)}} \\
\land & \underbrace{\left( G_{[10,11]}\left(\textit{way2} \Rightarrow F_{[3,4]} \textit{goal}\right) \land F_{[13,14]} G_{[0,1]} \textit{goal} \right)}_{\text{goal reaching (\textbf{nested} formula)}}
\end{aligned}
\end{equation}
requiring the vehicle to visit \textit{way1} between $t \in [5,7]$ and \textit{way2} between $t \in [8,11]$, avoid both obstacles over the full horizon, and reach \textit{goal} within 4 seconds of passing through \textit{way2}, and stabilize in \textit{goal}. Reachable tubes are computed using DynIbex\footnote{perso.ensta.fr/~chapoutot/dynibex/}~\cite{alexandre_dit_sandretto_validated_2016} with validated Runge-Kutta integration. The LP is solved using GLPK, and the confidence grid uses $N_A = 84$ points.

\subsection{Bounded disturbances} 

The first scenario considers bounded uncertainties: $[K_v] = [2, 2.2]$, $[K_c] = [1.8, 2]$, and $[P_v] = [-0.025, 0.025] = 2[P_c]$. Although these perturbations may vary over time, validated numerical integration requires them to be treated as constant within each integration step~\cite{alexandre_dit_sandretto_validated_2016}. 

   \begin{figure}[htbp]
       \begin{center}
        \includegraphics[width=\linewidth, trim=5 40 38 20, clip]{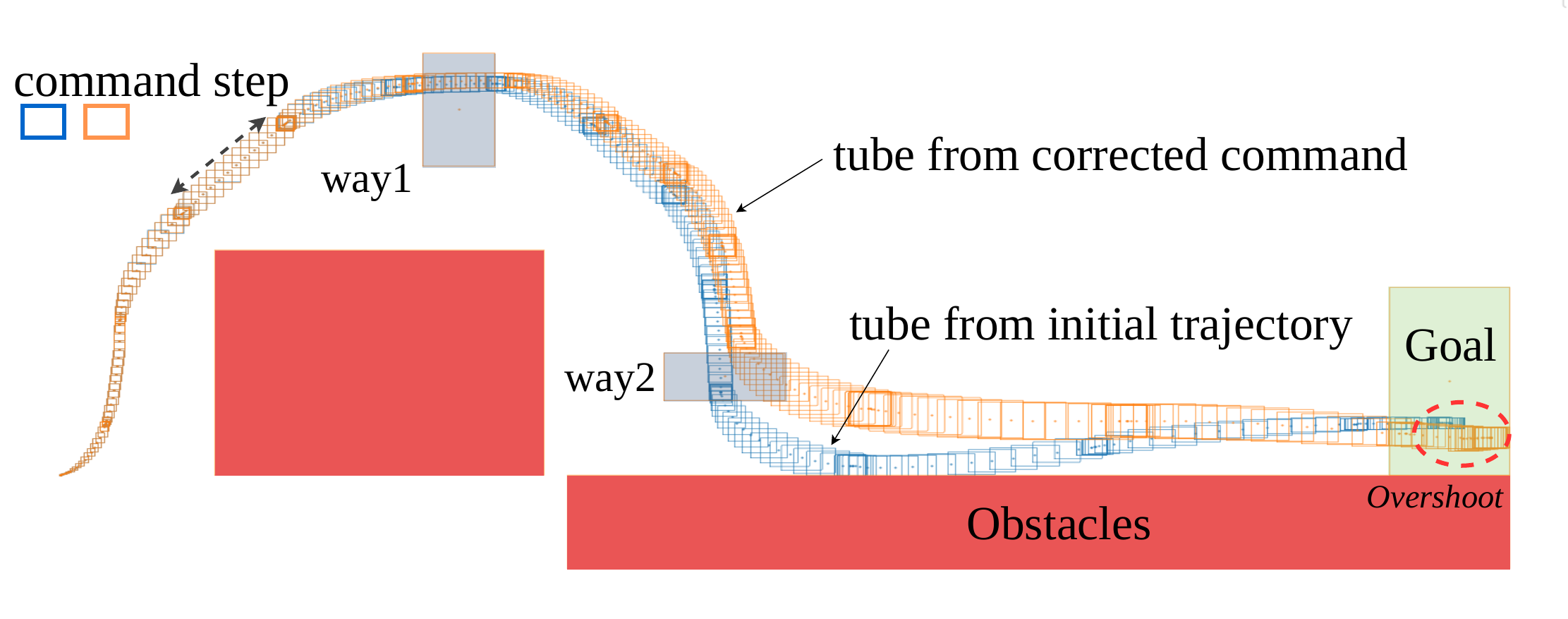}
      \caption{Initial tube in blue vs tube with corrected command in orange.}
      \label{boundeddisturb}
      \end{center}
   \end{figure}

The result is shown in Fig.~\ref{boundeddisturb}. The blue tube (no correction) intersects \textit{obstacle~2}, triggering the correction. The orange tube, obtained after solving Eq.~\eqref{eq:lp}, guarantees satisfaction of~$\phi$. The correction remains zero until deviation begins, then decelerates before \textit{waypoint~2} and accelerates toward the goal, producing a slight overshoot which can be reduced using finer piecewise steps.

\subsection{Probabilistic Threshold}

The second scenario (Fig.~\ref{probadisturb}, left) considers Gaussian uncertainties: $K_v \sim \mathcal{N}(2.1,\, 0.025^2)$, $K_c \sim \mathcal{N}(2.0,\, 0.025^2)$, and independent disturbances $P_v, P_c \sim \mathcal{N}(0,\, 0.04^2)$.
   
The reach sequence timing is modified from $\mathbf{G}_{[10,11]}$ to $\mathbf{G}_{[8,9]}$. The constraint computation correctly enforces exclusion over this period, preventing premature triggering of the reaching sequence. Due to the over-approximation $Z_\epsilon$, the deceleration before \textit{waypoint~2} is overly conservative. After this period, the system accelerates and stabilizes within the goal region.

The third scenario (Fig.~\ref{probadisturb}, right) is a multi-target case where uncertainty affects both waypoints and the stabilization in the goal region. The correction accelerates and steers the system to the right over a single correction step, after which the feedback rejoins the original trajectory and the system stabilizes correctly within the goal for 1\,s.

Fig.~\ref{sampleproba} displays the minimal corrections computed over a grid of 84 different parameter sets $S_{0.95}$ (thus $84$ scaling vectors), satisfying $\Pr\bigl((\hat{y}, t) \models \phi\bigr) \geq 0.95$, revealing variability in the correction magnitude with a global minimum of $\sum_j |\delta U_j \, v_j| = 1.5$ and a maximum of $2$. This trade-off is of practical interest when computation time is not critical.

\begin{figure}[htbp]
       \begin{center}
        \includegraphics[width=\linewidth, trim=10 2 0 10, clip]{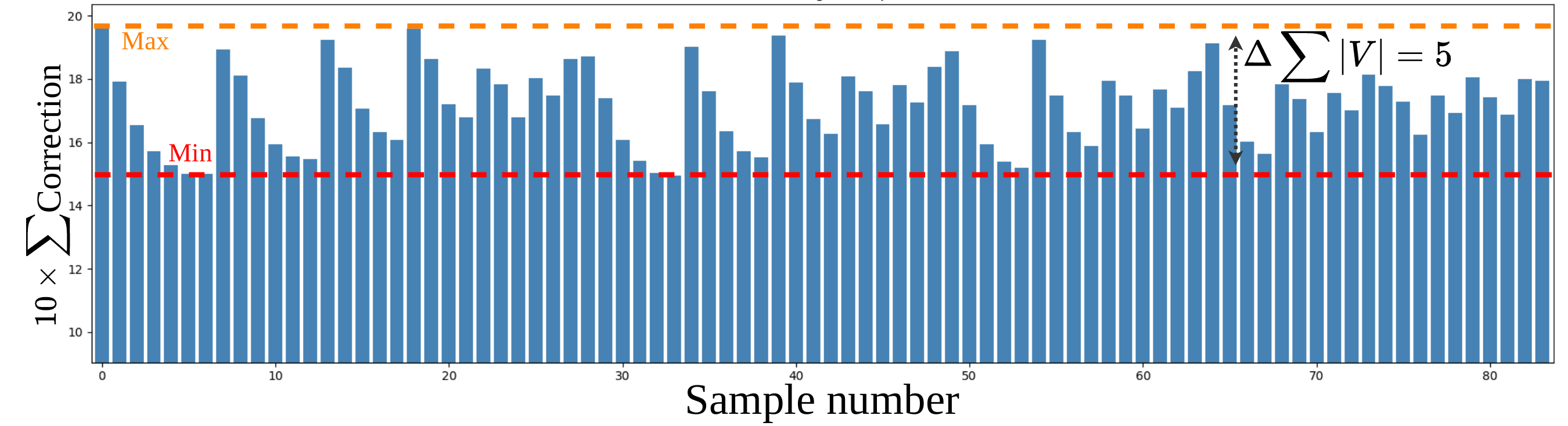}
      \caption{Sample number on the confidence grid vs sum of corrected command.}
      \label{sampleproba}
      \end{center}
\vspace{-3.5ex}
\end{figure}

\subsection{Computation Time}
The following computation times are on the formula horizon of 15s (using an HP laptop, Intel i7  and \textit{std::chrono})~\footnote{github.com/Antoine-Bst/STL-Three-Valued-Clause-Filtering/}.

Table~\ref{tab:lp_time} reports the computation time for a single $\mathcal{C}_k$.

\begin{table}[h]
\centering
\caption{Control correction computation for one $\mathcal{C}_k$ clause}
\label{tab:controlcomputation}
\begin{tabular}{l|c|c}
    Formula type & Number of constraints & LP computation time\\
    \hline
     Flat STL & 610 & $27 \pm 3$ ms\\
     Nested STL & 652 & $28 \pm 3$ ms\\
\end{tabular}
\label{tab:lp_time}
\vspace{-1.5ex}
\end{table}

The total correction time scales linearly with the number of clauses and remains practical. Since each clause guarantees satisfaction, dominated clauses under disjunction are pruned using a set-membership scoring function, reducing the computation to a single $\mathcal{C}_k$. Each clause is scored by aggregating predicate membership degrees with a minimum operator, consistent with conjunction semantics. This yields an increase in the average minimum correction across all samples (Fig.~\ref{sampleproba}) from 1.7 to 2.2, highlighting a trade-off between optimality and computation time.

\begin{table}[h]
\centering
\caption{Comparison of constraint identification methods (C++14)}
\label{tab:sat_comparison}
\begin{tabular}{llcc}
\hline
 & & \multicolumn{2}{c}{\textbf{Time}} \\
\textbf{Computation} & \textbf{Type} & \textbf{SAT} & \textbf{Ours} \\
\hline
\multirow{2}{*}{$\psi$ + $\mathrm{DNF}(\psi)$}
    & Flat ($K{=}13$)    & 38\,s            & $\approx 1$\,ms \\
    & Nested ($K{=}208$) & \textsc{t.o.}    & $\approx 1$\,ms \\
\hline
\multirow{2}{*}{$\mathrm{Cst}(\phi, \mathcal{C}_k)$}
    & Flat               & $48{\pm}3$\,ms   & $21{\pm}2$\,ms \\
    & Nested             & $119{\pm}3$\,ms  & $26{\pm}2$\,ms \\
\hline
\multirow{2}{*}{$\mathrm{DNF}(\mathrm{Cst})$}
    & Flat ($K{=}13$)              & \textsc{t.o.}    & $\approx 1$\,ms \\
    & Nested ($K{=}13$)           & \textsc{t.o.}    & $\approx 1$\,ms \\
\hline
\end{tabular}

\vspace{-2.5ex}
\end{table}

The number of clauses $\mathcal{C}_k$ in the DNF is denoted by $K$ in Table~\ref{tab:sat_comparison}.
The up to $26\times$ increase in computation time for propagation in $\mathrm{Cst}(\phi, \mathcal{C}_k)$ compared to the initial propagation and DNF is due to the [0,1] switching step: the number of predicate clauses to propagate and compute linear constraints grows from $40$--$55$ to nearly $700$. 

The implemented SAT Evaluation enumerates minimal satisfying conjunctions by iterative deepening over clause subsets, analogous to All-SAT with minimal subset filtering. Although it has access to the markers identified via~\cite{besset2025cdc} and their expected value, it has no prior on their logical dependencies within the formula. The combinatorial enumeration becomes intractable as the number of markers grows (Timeout $\textsc{t.o.}$ $>1$ min).
For $\mathrm{Cst}(\phi, \mathcal{C}_k)$, constraints are identified by backward testing ordered by robustness, where each constraint is tested for necessity which is fast to obtain. 

\vspace{-1ex}
\section{Conclusion and Discussion}\label{sec:conclusion}

This paper introduces a method for identifying and enforcing sufficient constraints to guarantee STL satisfaction under uncertainty. To our knowledge, no existing method combines tube-based guaranteed satisfaction assessment with corrective control. While computing the globally optimal correction is intractable in practice, our method provides a tractable and scalable alternative with formal guarantees, selecting a logically minimal sufficient clause subset at the cost of suboptimality. As an application, we instantiated this formulation in a control correction setting using zonotopic reachable sets and linear programming, showing how the extracted clauses recover guaranteed satisfaction.

While mandatory for online monitoring, the main bottleneck remains reachability analysis (5–10\,s), whereas the correction layer runs in milliseconds. Real-time deployment on a robotic platform therefore hinges primarily on accelerating this reachability step. Since our framework is independent of the reachability method, such improvements integrate directly.
A promising direction for future work is to investigate alternative approaches for computing control corrections, as iterations due to tight control bounds (Section~\ref{sec:feasibility}) arise frequently.
\vspace{-3ex}

\section*{Acknowledgment}

The authors$^a$ acknowledge support from the CIEDS\footnote{CIEDS: French Interdisciplinary Center for Defense and Security.} with the STARTS project.
\vspace{-1ex}
{\tiny
\bibliographystyle{IEEEtran}
\bibliography{refs}

@article{capdarticle,
title = {CAPD::DynSys: A flexible C++ toolbox for rigorous numerical analysis of dynamical systems},
journal = {Communications in Nonlinear Science and Numerical Simulation},
year = {2021},
author = {Tomasz Kapela and Marian Mrozek and Daniel Wilczak and Piotr Zgliczyński}
}

@inproceedings{maler_monitoring_2004,
    author = {Maler, Oded and Nickovic, Dejan},
    title = {Monitoring Temporal Properties of Continuous Signals},
    booktitle = {Formal Techniques, Modelling and Analysis of Timed and Fault-Tolerant Systems},
    publisher = {Springer},
    year = {2004},
    doi = {10.1007/978-3-540-30206-3\_12}
}

@incollection{jaulin_applied_2001,
	location = {London},
	title = {Applied Interval Analysis},
	booktitle = {Applied Interval Analysis},
	year = {2001},
	publisher = {Springer},
	author = {Jaulin, Luc and Kieffer, Michel and Didrit, Olivier and Walter, Eric},
	date = {2001},
	langid = {english},
}

@article{alexandre_dit_sandretto_validated_2016,
  author       = {Julien {Alexandre~dit~Sandretto} and Alexandre Chapoutot},
  title        = {Validated Explicit and Implicit Runge-Kutta Methods},
  journal      = {Reliable Computing},
  volume       = {22},
  year         = {2016},
}

@inproceedings{Althoff2015ARCH,
	author			= {Matthias Althoff},
	title			= {An Introduction to {CORA} 2015},
	year			= {2015},
	month 			= {December},
	booktitle		= {Proc. of the 1st and 2nd Workshop on Applied Verification for Continuous and Hybrid Systems},
	doi			= {10.29007/zbkv},
  }

@inproceedings{Lercher2024CAV,
  author       = {Florian Lercher and Matthias Althoff},
  title        = {Using Four-Valued Signal Temporal Logic for Incremental Verification of Hybrid Systems},
  booktitle    = {Proc. of Computer Aided Verification (CAV)},
  series       = {Lecture Notes in CS},
  year         = {2024},
  doi          = {10.1007/978-3-031-65633-0\_12},
}

@article{Ishii2016IEICE,
  author       = {Daisuke Ishii and Naoki Yonezaki and Alexandre Goldsztejn},
  title        = {Monitoring Temporal Properties Using Interval Analysis},
  journal      = {IEICE Transactions on Fundamentals of Electronics, Communications and Computer Sciences},
  year         = {2016},
  doi          = {10.1587/transfun.E99.A.442},
}

@inproceedings{Roehm2016ATVA,
  author    = {Hendrik Roehm and Jens Oehlerking and Thomas Heinz and Matthias Althoff},
  title     = {{STL} Model Checking of Continuous and Hybrid Systems},
  booktitle = {Automated Technology for Verification and Analysis (ATVA 2016)},
  year      = {2016},
  doi       = {10.1007/978-3-319-46520-3\_26},
}

@inproceedings{SankarMonteCarlo,
author = {Nghiem, Truong and Sankaranarayanan, Sriram and Fainekos, Georgios and Ivanci\'{c}, Franjo and Gupta, Aarti and Pappas, George J.},
title = {Monte-carlo techniques for falsification of temporal properties of non-linear hybrid systems},
year = {2010},
doi = {10.1145/1755952.1755983},
booktitle = {Proceedings 13th ACM International Conference on Hybrid Systems: Computation and Control}
}

@article{alexandre_dit_sandretto_confidence_2021,
  author       = {Julien {Alexandre~dit~Sandretto}},
  title        = {Confidence-based Contractor, Propagation and Potential Clouds for Differential Equations},
  journal      = {Acta Cybernetica},
  year         = {2021},
  doi          = {10.14232/actacyb.285177}
}

@inproceedings{sadigh2016safe,
  author    = {Dorsa Sadigh and Ashish Kapoor},
  title     = {Safe Control under Uncertainty with Probabilistic Signal Temporal Logic},
  booktitle = {Proceedings of Robotics: Science and Systems (RSS)},
  year      = {2016}
}

@inproceedings{besset2025cdc,
  author    = {Antoine Besset and Joris Tillet and Julien {Alexandre~dit~Sandretto}},
  title     = {Uncertainty Removal in Verification of Nonlinear Systems against Signal Temporal Logic via Incremental Reachability Analysis},
  booktitle = {Proceedings of the 64th IEEE CDC},
  year      = {2025},
  publisher = {IEEE},
}

@inproceedings{Tran2025ProbStarTL,
  author    = {Hoang{-}Dung Tran and Sung Woo Choi and Yuntao Li and Hideki Okamoto and Bardh Hoxha and Georgios Fainekos},
  title     = {ProbStar Temporal Logic for Verifying Complex Behaviors of Learning-enabled Systems},
  booktitle = {Proceedings of the 28th ACM International Conference on Hybrid Systems: Computation and Control},
  year      = {2025},
  doi       = {10.1145/3716863.3718035},
}

@inproceedings{lindemann2023conformal,
  title={Conformal prediction for stl runtime verification},
  author={Lindemann, Lars and Qin, Xin and Deshmukh, Jyotirmoy V and Pappas, George J},
  booktitle={Proceedings ACM/IEEE 14th International Conference on Cyber-Physical Systems},
  year={2023}
}

@inproceedings{Shmarov2015ProbReach,
  title     = {ProbReach: A Tool for Guaranteed Reachability Analysis of Stochastic Hybrid Systems},
  author    = {Shmarov, Fedor and Zuliani, Paolo},
  booktitle = {Symbolic and Numerical Methods for Reachability Analysis (SNR 2015)},
  year      = {2015},
}

@ARTICLE{Farahani2018,
  author={Farahani, Samira S. and Majumdar, Rupak and Prabhu, Vinayak S. and Soudjani, Sadegh},
  journal={IEEE Transactions on Automatic Control}, 
  title={Shrinking Horizon Model Predictive Control With Signal Temporal Logic Constraints Under Stochastic Disturbances}, 
  year={2019},
  doi={10.1109/TAC.2018.2880651}}

@inproceedings{vlahakis2024probabilistic,
  title={Probabilistic tube-based control synthesis of stochastic multi-agent systems under signal temporal logic},
  author={Vlahakis, Eleftherios E and Lindemann, Lars and Sopasakis, Pantelis and Dimarogonas, Dimos V},
  booktitle={2024 IEEE 63rd Conference on Decision and Control (CDC)},
  year={2024},
  organization={IEEE}
}

@ARTICLE{shurmannzonocontrol2021,
  author={Schürmann, Bastian and Althoff, Matthias},
  journal={IEEE Transactions on Automatic Control}, 
  title={Optimizing Sets of Solutions for Controlling Constrained Nonlinear Systems}, 
  year={2021},
  doi={10.1109/TAC.2020.2989762}}

@article{SCHURMANN2017,
title = {Guaranteeing Constraints of Disturbed Nonlinear Systems Using Set-Based Optimal Control in Generator Space},
year = {2017},
note = {20th IFAC World Congress},
issn = {2405-8963},
doi = {https://doi.org/10.1016/j.ifacol.2017.08.1617},
author = {Bastian {Schürmann} and Matthias Althoff},
}

@article{safeautonomy2018,
author = {Jha, Susmit and Raman, Vasumathi and Sadigh, Dorsa and Seshia, Sanjit A.},
title = {Safe Autonomy Under Perception Uncertainty Using Chance-Constrained Temporal Logic},
year = {2018},
issn = {0168-7433},
doi = {10.1007/s10817-017-9413-9},
journal = {J. Autom. Reason.}
}

@article{lindemann2021reactive,
  title={Reactive and risk-aware control for signal temporal logic},
  author={Lindemann, Lars and Pappas, George J and Dimarogonas, Dimos V},
  journal={IEEE Transactions on Automatic Control},
  year={2021},
  publisher={IEEE}
}

@article{majumdar2017funnel,
  author    = {Anirudha Majumdar and Russ Tedrake},
  title     = {Funnel Libraries for Real-Time Robust Feedback Motion Planning},
  journal   = {International Journal of Robotics Research},
  year      = {2017},
  doi       = {10.1177/0278364917712421}
}

@INPROCEEDINGS{dawson2022,
  author={Dawson, Charles and Fan, Chuchu},
  booktitle={2022 IEEE/RSJ International Conference on Intelligent Robots and Systems}, 
  title={Robust Counterexample-guided Optimization for Planning from Differentiable Temporal Logic}, 
  year={2022},
  doi={10.1109/IROS47612.2022.9981382}}

@InProceedings{meng2025telograf,
  title = 	 {{T}e{L}o{G}ra{F}: Temporal Logic Planning via Graph-encoded Flow Matching},
  author =       {Meng, Yue and Fan, Chuchu},
  booktitle = 	 {Proceedings of the 42nd International Conference on Machine Learning},
  year = 	 {2025}
}

@article{Yuonline2024,
author = {Pian Yu and Yulong Gao and Frank J. Jiang and Karl H. Johansson and Dimos V. Dimarogonas},
title ={Online control synthesis for uncertain systems under signal temporal logic specifications},
journal = {International Journal of Robotics Research},
year = {2024},
doi = {10.1177/02783649231212572},
}

@INPROCEEDINGS{verhagen2024,
  author={Verhagen, Joris and Lindemann, Lars and Tumova, Jana},
  booktitle={2024 IEEE 63rd Conference on Decision and Control (CDC)}, 
  title={Robust STL Control Synthesis under Maximal Disturbance Sets}, 
  year={2024},
  doi={10.1109/CDC56724.2024.10886238}}

@inproceedings{ferrereimplicant,
title = "Trace Diagnostics Using Temporal Implicants",
author = "Dejan Nickovic and Thomas Ferr{\`e}re and Oded Maler",
year = "2015",
language = "English",
booktitle = "Automated Technology for Verification and Analysis - 13th International Symposium, ATVA",
}
}

\end{document}